\documentclass[11pt]{article}

\usepackage{a4}
\usepackage[top=30truemm,bottom=30truemm]{geometry}
\usepackage{authblk}
\usepackage{amsthm}
\usepackage{amsmath,amssymb}
\usepackage{tabls}
\usepackage[dvipdfmx]{graphicx}
\usepackage{array}
\usepackage{url}
\usepackage{cases}
\usepackage{wrapfig}
\usepackage[normalem]{ulem}

\usepackage[final,mode=multiuser]{fixme}
\fxuseenvlayout{colorsig}
\fxusetargetlayout{color}
\FXRegisterAuthor{nf}{anf}{NF}
\FXRegisterAuthor{yn}{ayn}{YN}
\FXRegisterAuthor{si}{asi}{SI}
\FXRegisterAuthor{hb}{ahb}{HB}
\FXRegisterAuthor{mt}{amt}{MT}
\fxsetface{env}{}

\newtheorem{theorem}{Theorem}
\newtheorem{lemma}{Lemma}

\newtheorem{proposition}{Proposition}
\newtheorem{corollary}{Corollary}
\theoremstyle{definition}
\newtheorem{definition}{Definition}

\newcommand{\Prefix}{\mathsf{Prefix}}
\newcommand{\Substr}{\mathsf{Substr}}
\newcommand{\Suffix}{\mathsf{Suffix}}

\newcommand{\prev}{\mathsf{prev}}
\newcommand{\SHT}{\mathsf{SHT}}
\newcommand{\PPH}{\mathsf{PPH}}

\newcommand{\id}{\mathsf{id}}

\newcommand{\slink}{\mathsf{sl}}
\newcommand{\rslink}{\mathsf{rsl}}
\newcommand{\mrp}{\mathsf{mrp}}

\newcommand{\occ}{\mathit{occ}}
\newcommand{\pocc}{\mathit{pocc}}

\title{
  Right-to-left online construction of \\
  parameterized position heaps
}

\author{Noriki~Fujisato}
\author{Yuto~Nakashima}
\author{Shunsuke~Inenaga}
\author{Hideo~Bannai}
\author{Masayuki~Takeda}

\affil{\textit{Department of Informatics, Kyushu University, Japan} \\
\texttt{\small \{noriki.fujisato, yuto.nakashima, inenaga, bannai, takeda\}@inf.kyushu-u.ac.jp}
}

\date{}

\begin{document}
\maketitle

\begin{abstract}
	Two strings of equal length are said to \emph{parameterized match}
	if there is a bijection that maps the characters of one string to those
	of the other string, so that two strings become identical.
	The parameterized pattern matching problem is,
	given two strings $T$ and $P$, to find
	the occurrences of substrings in $T$ that parameterized match $P$.
	Diptarama et al. [Position Heaps for Parameterized Strings, CPM 2017]
	proposed an indexing data structure called \emph{parameterized position heaps},
	and gave a left-to-right online construction algorithm.
	In this paper, we present a \emph{right-to-left} online construction
	algorithm for parameterized position heaps.
	For a text string $T$ of length $n$ over two kinds of alphabets
	$\Sigma$ and $\Pi$ of respective size $\sigma$ and $\pi$,
	our construction algorithm runs in $O(n \log(\sigma + \pi))$ time
	with $O(n)$ space.
	Our right-to-left parameterized position heaps support
	pattern matching queries in $O(m \log (\sigma + \pi) + m \pi + \pocc))$ time,
	where $m$ is the length of a query pattern $P$
	and $\pocc$ is the number of occurrences to report.
	Our construction and pattern matching algorithms are
	as efficient as Diptarama et al.'s algorithms.
\end{abstract}

\section{Introduction}

\emph{Text indexing} is the task to preprocess the text string
so that subsequent pattern matching queries can be answered efficiently.
To date, a numerous number of text indexing structure for exact pattern matching
have been proposed, ranging from classical data structures such as suffix trees~\cite{Weiner},
directed acyclic word graphs~\cite{blumer85:_small_autom_recog_subwor_text,Blumer87},
and suffix arrays~\cite{manber93:_suffix},
to more advanced ones such as
compressed suffix arrays~\cite{GrossiV05} and FM index~\cite{FerraginaM05},
just to mention a few.

Ehrenfeucht et al.~\cite{ehrenfeucht_position_heaps_2011}
proposed a text indexing structure called \emph{position heaps}.
Ehrenfeucht et al.'s position heap is constructed in a \emph{right-to-left} online manner,
where a new node is incrementally inserted to the current position heap
for each decreasing position $i = n, \ldots, 1$ in the input string $T$ of length $n$.
In other words, Ehrenfeucht et al.'s position heap
is defined over a sequence $\langle \varepsilon, T[n..], \ldots, T[1..] \rangle$
of the suffixes of $T$ in increasing order of their length,
where $\varepsilon$ is the empty string of length $0$.
Kucherov~\cite{Kucherov13} proposed another variant of position heaps.
Kucherov's position heap is constructed in a \emph{left-to-right} online manner,
where a new node is incrementally inserted to the current position heap
for each increasing $i = 1, \ldots, n$.
In other words, Kucherov's position heap is defined over
a sequence $\langle T[1..], \ldots, T[n..], \varepsilon \rangle$
of the suffixes of $T$ in decreasing order of their length.
We will call Ehrenfeucht et al.'s position heap as the RL position heap,
and Kucherov's position heap as the LR position heap.
Both of the RL and LR position heaps for a text string $T$
of length $n$ require $O(n)$ space and can be constructed in $O(n \log \sigma)$ time,
where $\sigma$ is the alphabet size.
By augmenting the RL and LR position heaps of $T$
with auxiliary links called maximal reach pointers,
pattern matching queries can be answered in $O(m \log \sigma + \occ)$ time,
where $m$ is the length of a query pattern $P$
and $\occ$ is the number of occurrences of $P$ in $T$.

Nakashima et al.~\cite{position_heaps_of_trie_2012}
proposed position heaps for a set of strings
that is given as a reversed trie,
and proposed an algorithm that constructs
the position heap of a given trie in $O(\sigma N)$ time and space,
where $N$ is the size of the input trie.
Later, the same authors
showed how to construct the position heap
of a trie in $O(N)$ time and space,
for integer alphabets of size polynomialy bounded in $N$~\cite{NakashimaIIBT15}.

Baker~\cite{Baker96} introduced the \emph{parameterized pattern matching} problem,
that seeks for the occurrences of substrings of the text $T$
that have the ``same'' structures as the given pattern $P$.
Parameterized pattern matching is motivated
by e.g., software maintenance and plagiarism detection~\cite{Baker96}.
More formally, we consider two distinct alphabets $\Sigma$ and $\Pi$,
and we call an element over $\Sigma \cup \Pi$ a p-string.
The parameterized pattern matching problem is,
given two p-strings $T$ and $P$,
to find all occurrences of substrings $X$ of $T$
that can be transformed to $P$ by a bijection from $\Sigma \cup \Pi$ to $\Sigma \cup \Pi$
which is identity for $\Sigma$.
For instance,
if $T = \mathtt{abzaxxbyaxxbazzax}$ and $P = \mathtt{yazzbx}$
where $\Sigma = \{\mathtt{a, b}\}$
and $\Pi = \{\mathtt{x, y, z}\}$,
then the positions to output are $3$ and $8$.
To see why, observe that for the substring $T[3..8] = \mathtt{zaxxby}$
there is a bijection
$\mathtt{z} \rightarrow \mathtt{y}$,
$\mathtt{a} \rightarrow \mathtt{a}$,
$\mathtt{x} \rightarrow \mathtt{z}$,
$\mathtt{b} \rightarrow \mathtt{b}$, and
$\mathtt{x} \rightarrow \mathtt{y}$
that maps the substring to $P$.
Also, observe that for the other substring $T[8..13] = \mathtt{yaxxbz}$,
there is a bijection
$\mathtt{y} \rightarrow \mathtt{y}$,
$\mathtt{a} \rightarrow \mathtt{a}$,
$\mathtt{x} \rightarrow \mathtt{z}$,
$\mathtt{b} \rightarrow \mathtt{b}$, and
$\mathtt{z} \rightarrow \mathtt{x}$
that maps the substring to $P$ as well.

Of various algorithms and indexing structures
for the parameterized pattern matching (see~\cite{MendivelsoP15} for a survey),
we focus on Diptarama et al.'s \emph{parameterized position heaps}~\cite{DiptaramaKONS17}.
Diptarama et al.'s \emph{parameterized position heaps} are based
on Kucherov's LR position heaps,
which are constructed in a \emph{left-to-right} online manner.
Let us call their structure the \emph{LR p-position heaps}.
Diptarama et al. showed how to construct the LR p-position heap
for a given text of length $n$ in $O(n \log (\sigma + \pi))$ time
with $O(n)$ space, where $\sigma = |\Sigma|$ and $\pi = |\Pi|$.
They also showed that the LR p-position heap
augmented with maximal reach pointers can support
parameterized pattern matching queries in $O(m \log (\sigma + \pi) + m\pi + \pocc)$ time,
where $\pocc$ is the number of occurrences to report.

In this paper, we propose \emph{RL p-position heaps}
which are constructed in a \emph{right-to-left} online manner.
We show how to construct our RL position heap for a given text string $T$
of length $n$ in $O(n \log (\sigma + \pi))$ time with $O(n)$ space.
Our construction algorithm is based on
Ehrenfeucht et al.'s construction algorithm for
RL position heaps~\cite{ehrenfeucht_position_heaps_2011},
and Weiner's suffix tree construction algorithm~\cite{Weiner}.
Namely, we use reversed suffix links
defined for the nodes of RL p-position heaps.
The key to our algorithm is how to label the reversed suffix links,
which will be clarified in Definition~\ref{def:reversed_suffix_link}.
Using our RL p-position heap augmented with maximal reach pointers,
one can perform parameterized pattern matching queries
in $O(m \log (\sigma + \pi) + m\pi + \pocc)$ time.

\section{Preliminaries}

\subsection{Notations on strings}

Let $\Sigma$ and $\Pi$ be disjoint sets
called a \emph{static alphabet} and a \emph{parameterized alphabet},
respectively.
Let $\sigma = |\Sigma|$ and $\pi = |\Pi|$.
An element of $\Sigma$ is called an \emph{s-character},
and that of $\Pi$ is called a \emph{p-character}.
In the sequel, both an s-character and a p-character
are sometimes simply called a \emph{character}.
An element of $\Sigma^*$ is called a \emph{string},
and an element of $(\Sigma \cup \Pi)^*$ is called a \emph{p-string}.
The length of a (p-)string $S$ is the number of
characters contained in $S$.
The empty string $\varepsilon$ is a string of length 0,
namely, $|\varepsilon| = 0$.
For a (p-)string $S = XYZ$, $X$, $Y$ and $Z$ are called
a \emph{prefix}, \emph{substring}, and \emph{suffix} of $w$, respectively.
The set of prefixes, substrings, and suffixes
of a (p-)string $S$ is denoted by $\Prefix(S)$, $\Substr(S)$, and $\Suffix(S)$,
respectively.
The $i$-th character of a (p-)string $S$ is denoted by
$S[i]$ for $1 \leq i \leq |S|$,
and the substring of a (p-)string $S$ that begins at position $i$ and
ends at position $j$ is denoted by $S[i..j]$ for $1 \leq i \leq j \leq |S|$.
For convenience, let $S[i..j] = \varepsilon$ if $j < i$.
Also, let $S[i..] = S[i..|S|]$ for any $1 \leq i \leq |S|$.

\subsection{Parameterized pattern matching}

For any p-string $X$ and $f:(\Sigma \cup \Pi) \rightarrow (\Sigma \cup \Pi)$,
let $F(X) = f(X[1]) \cdots f(X[|X|])$.
Two p-strings $X$ and $Y$ of length $k$ each
are said to \emph{parameterized match} (\emph{p-match})
iff there is a bijection $f$ on $\Sigma \cup \Pi$
such that $f(a) = a$ for any $a \in \Sigma$
and $f(X[i]) = Y[i]$ for all $1 \leq i \leq k$.
For instance, if $\Sigma = \{\mathtt{a}, \mathtt{b}\}$ and $\Pi = \{\mathtt{x}, \mathtt{y}, \mathtt{z}\}$,
then $X = \mathtt{axbzzayx}$ and $Y = \mathtt{azbyyaxz}$ p-match
since there is a bijection $f$
such that $f(\mathtt{a}) = \mathtt{a}$, $f(\mathtt{b}) = \mathtt{b}$,
$f(\mathtt{x}) = \mathtt{z}$, $f(\mathtt{y}) = \mathtt{x}$, and $f(\mathtt{z}) = \mathtt{y}$
and $F(X) = F(\mathtt{axbzzayx}) = \mathtt{azbyyaxz} = Y$.
We write $X \approx Y$ iff $X$ and $Y$ p-match.

The \emph{previous encoding} $\prev(S)$ of a p-string $S$ of length $n$
is a sequence of length $n$ such that
the first occurrence of each p-character $x$
is replaced with $0$ and any other occurrence of $x$ is
replaced by the distance to the previous occurrence of $x$ in $S$,
and each s-character remains the same.
More formally, $\prev(S)$ is a sequence over $\Sigma \cup [0..n-1]$
of length $n$ such that for each $1 \leq i \leq n$,
\[
	\prev(S)[i] =
	\begin{cases}
		S[i] & \mbox{if } S[i] \in \Sigma,                                                                  \\
		0    & \mbox{if } S[i] \in \Pi \mbox{ and } S[i] \neq S[j] \mbox{ for any } 1 \leq j < i,           \\
		i-j  & \mbox{if } S[i] \in \Pi, S[i] = S[j] \mbox{ and } S[i] \neq S[k] \mbox{ for any } j < k < i.
	\end{cases}
\]

Observe that $X \approx Y$ iff $\prev(X) = \prev(Y)$.
Using the same example as above,
we have that $\prev(\mathtt{axbzzayx}) = \prev(\mathtt{azbyyaxz}) = \mathtt{a}0\mathtt{b}01\mathtt{a}06$.

Let $T$ and $P$ be p-strings of length $n$ and $m$, respectively,
where $n \geq m$.
The \emph{parameterized pattern matching} problem is
to find all positions $i$ in $T$ such that $T[i..i+m-1] \approx P$.

\section{Parameterized position heaps}

Let $\mathcal{S} = \langle S_1, \ldots, S_k \rangle$ be a sequence of strings
such that for any $1 < i \leq k$, 
$S_i \not \in \Prefix(S_j)$ for any $1 \leq j < i$.
For convenience, we assume that $S_1 = \varepsilon$.

\begin{definition}[Sequence hash trees~\cite{coffman}]
\label{def:seq_hash}
The \emph{sequence hash tree} of a sequence
$\mathcal{S} = \langle S_1, \ldots, S_k \rangle$ of strings,
denoted $\SHT(\mathcal{S})$, is a trie structure that is recursively defined as follows:
Let $\SHT(\mathcal{S})^{i} = (V_i, E_i)$.
Then 
\[
 \SHT(S)^{i} = 
  \begin{cases}
   (\{\varepsilon\}, \emptyset) & \mbox{if $i = 1$}, \\
   (V_{i-1} \cup \{p_{i}\}, E_{i-1} \cup \{(q_{i}, c, p_{i})\}) & \mbox{if $1 \leq i \leq k$},
  \end{cases}
\]
where $q_{i}$ is the longest prefix of $S_i$ which satisfies $q_{i} \in V_{i-1}$,
$c = S_i[|q_{i}|+1]$,
and $p_{i}$ is the shortest prefix of $S_i$ which satisfies $p_{i} \notin V_{i-1}$.
\end{definition}
Note that since we have assumed that each $S_i \in \mathcal{S}$ is not a prefix of 
$S_j$ for any $1 \leq j < i$, 
the new node $p_i$ and new edge $(q_{i}, c, p_{i})$ always exist 
for each $1 \leq i \leq k$.
Clearly $\SHT(\mathcal{S})$ contains $k$ nodes (including the root).

In what follows,
we will define our indexing data structure for a text p-string $T$ of length $n$.
Let $\mathcal{P}_T = \langle \varepsilon, \prev(T[n..]), \ldots, \prev(T[1..]) \rangle$
be the sequence of previous encoded suffixes of $T$
arranged in \emph{increasing} order of their length.
It is clear that $\prev(T[i..]) \notin \Prefix(\prev(T[j..]))$ for any $1 \leq j < i$
and $\prev(T[i..]) \notin \Prefix(\varepsilon)$
for any $1 \leq i \leq n$.
Hence we can naturally define the sequence hash tree for
$\mathcal{P}_T$, and we obtain our data structure:
\begin{definition}[Parameterized positions heaps]
\label{def:p_position_heap}
The \emph{parameterized position heap} (\emph{p-position heap})
for a p-string $T$, 
denoted $\PPH(T)$, is the sequence hash tree of $\mathcal{P}_T$
i.e., $\PPH(T) = \SHT(\mathcal{P}_T)$.
\end{definition}

\begin{figure}[tb]
   \centerline{
    \begin{tabular}{|l|r|} \hline
     $\prev(T[17..])$ & $\underline{0}$ \\ \hline
     $\prev(T[16..])$ & $\underline{00}$ \\ \hline
     $\prev(T[15..])$ & $\underline{\mathtt{a}}00$ \\ \hline
     $\prev(T[14..])$ & $\underline{0\mathtt{a}}03$ \\ \hline
     $\prev(T[13..])$ & $\underline{00\mathtt{a}}33$ \\ \hline
     $\prev(T[12..])$ & $\underline{000}\mathtt{a}33$ \\ \hline
     $\prev(T[11..])$ & $\underline{01}00\mathtt{a}33$ \\ \hline
     $\prev(T[10..])$ & $\underline{001}04\mathtt{a}33$ \\ \hline
     $\prev(T[9..])$ & $\underline{010}104\mathtt{a}33$ \\ \hline
     $\prev(T[8..])$ & $\underline{0013}104\mathtt{a}33$ \\ \hline
     $\prev(T[7..])$ & $\underline{0101}3104\mathtt{a}33$ \\ \hline
     $\prev(T[6..])$ & $\underline{00131}3104\mathtt{a}33$ \\ \hline
     $\prev(T[5..])$ & $\underline{01013}13104\mathtt{a}33$ \\ \hline
     $\prev(T[4..])$ & $\underline{001313}13104\mathtt{a}33$ \\ \hline
     $\prev(T[3..])$ & $\underline{002}131313104\mathtt{a}33$ \\ \hline
     $\prev(T[2..])$ & $\underline{0022}131313104\mathtt{a}33$ \\ \hline
     $\prev(T[1..])$ & $\underline{\mathtt{a}0}022131313104\mathtt{a}33$ \\ \hline       
    \end{tabular}
    \hfill
    \raisebox{-27mm}{\includegraphics[scale=0.65]{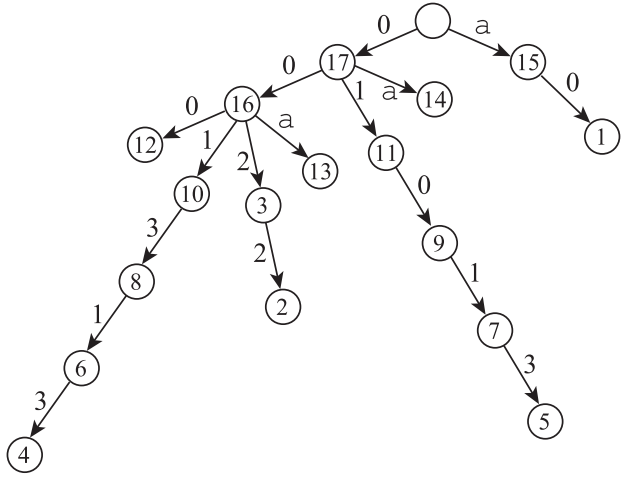}}
  }
   \caption{To the left is the list of $\prev(T[i..])$
     for p-string $T = \mathtt{axyxyyxxyyxxzyazy}$ of length $17$,
     where $\Sigma = \{\mathtt{a}\}$ and $\Pi = \{\mathtt{x}, \mathtt{y}, \mathtt{z}\}$.
     To the right is an illustration for $\PPH(T)$.
     The underlined prefix of each $\prev(T[i..])$ in the left list
     denotes the longest prefix of $\prev(T[i..])$ that was inserted to
     $\PPH(T[i+1..])$ and hence, the node with id $i$
     represents this underlined prefix of $\prev(T[i..])$.
}
  \label{fig:p-position_heap}
\end{figure}
See Figure~\ref{fig:p-position_heap} for an example of our p-position heap.

Note that we can obtain $\mathcal{P}_{T[i-1..]}$
by adding $\prev(T[i-1..])$ at the beginning of $\mathcal{P}_{T[i..]}$.
This also means that $\PPH(T[i..]) = \SHT(\mathcal{P}_{T[i..]})$
for each $1 \leq i \leq n$.
Hence, we can construct $\PPH(T)$
by processing the input string $T$ from right to left.
We remark that we can easily compute $\prev(T[i-1..])$ from $\prev(T[i..])$
in a total of $O(n \log \pi)$ time for all $2 \leq i \leq n$
using $O(\min\{\pi, n\})$ extra space,
e.g., by maintaining a balanced search tree
that stores the distinct p-characters that have occurred in $T[i..]$
and records the leftmost occurrences of these p-character in the nodes.

Diptarama et al.~\cite{DiptaramaKONS17}
proposed another version of parameterized position heap
for a sequence of previous encoded suffixes of the input p-string $T$
arranged in \emph{decreasing} order of their length.
Since their algorithm processes $T$ from left to right,
we sometimes call their structure as a \emph{left-to-right p-position heap}
(\emph{LR p-position heap}),
while we call our $\PPH(T)$ as a \emph{right-to-left p-position heap}
(\emph{RL p-position heap}) since our construction algorithm
processes $T$ from right to left.

For any p-string $P \in (\Sigma \cup [0..n-1])^+$,
we say that $P$ is \emph{represented} by $\PPH(T)$
iff $\PPH(T)$ has a path which starts from the root and spells out $P$.

\begin{lemma} \label{lem:position_nodes_correspondence}
  For any string $T$ of length $n$,
  $\PPH(T)$ consists of exactly $n+1$ nodes.
  Also, there is a one-to-one correspondence between
  the positions $1, \ldots, n$ in $T$
  and the non-root nodes of $\PPH(T)$.
\end{lemma}

\begin{proof}
  Initially, $\PPH(\varepsilon)$ consists only of the root
  that represents $\varepsilon$.
  For each $1 \leq i \leq n$,
  since $|\prev(T[i..])| = n-i+1 > n-j+1 = |\prev(T[j..])|$
  for any $1 \leq i < j \leq n$,
  it is clear that there is a prefix of $\prev(T[i..])$
  that is not represented by $\PPH(T[i+1..])$.
  Therefore, when we construct $\PPH(T[i..])$ from $\PPH(T[i+1..])$,
  then exactly one node is inserted, which corresponds to position $i$.
\end{proof}

Let $V$ be the set nodes of $\PPH(T)$.
Based on Lemma~\ref{lem:position_nodes_correspondence},
we define a bijection $\id:V \rightarrow [0..n]$ such that
$\id(r) = 0$ for the root $r$ and 
$\id(v) = i$ iff $v$ was the node that was inserted
when constructing $\PPH(T[i..])$ from $\PPH(T[i+1..])$.

Unlike our RL p-position heap,
Diptarama et al.'s LR p-position heap can have
\emph{double nodes} to which two positions of the text p-string are associated.

We remark that the pattern matching algorithm of
Diptarama et al.~\cite{DiptaramaKONS17}
can be applied to our RL p-position heap $\PPH(T)$ for a text p-string $T$,
and this way one can solve the parameterized pattern matching problem
in $O(m \log (\sigma + \pi) + m\pi + \occ)$ time,
where $\occ$ is the number of positions in text $T$
such that the pattern p-string $P$ of length $m$
and the corresponding substring $T[i..i+m-1]$ p-match.
We note that since our RL p-position heap does not have double nodes,
the pattern matching algorithm can be somewhat simplified.

The following lemma is an analogue
to Lemma 6 of~\cite{DiptaramaKONS17} for
Diptarama et al.'s LR p-position heap.

\begin{lemma} \label{lem:substring_containment}
For any $1 \leq i \leq j \leq n$
if $\prev(T[i..j])$ is represented by $\PPH(T)$,
then for any substring $X$ of $T[i..j]$,
$\prev(X)$ is represented by $\PPH(T)$.
\end{lemma}

\begin{proof}
  The lemma can be shown in a similar way to Lemma 6 of~\cite{DiptaramaKONS17}.
  For the sake of completeness, we provide a full proof below.
  
  First, we show that for any proper prefix $T[i..i+k]$ of $T[i..j]$
  with $0 \leq k < j-i$,
  $\prev(T[i..i+k])$ is represented by $\PPH(T)$.
  It follows from the definition of previous encoding
  that $\prev(T[i..i+k]) = \prev(T[i..j])[1..k+1]$,
  and hence $\prev(T[i..i+k])$ is a prefix of $\prev(T[i..j])$.
  Since $\prev(T[i..j])$ is represented by $\PPH(T)$
  and $i \leq i + k < j$,
  $\prev(T[i..i+k])$ is also represented by $\PPH(T)$.

  Now it suffices for us to show that
  for any proper suffix $T[i+h..j]$ of $T[i..j]$ with $0 < h \leq j-i$,
  $\prev(T[i+h..j])$ is represented by $\PPH(T)$,
  since then we can inductively apply the above discussion for the prefixes.
  By the above discussions for the prefixes of $T[i..j]$,
  there exist positions $i = b_{j-i} < \cdots < b_0 \leq n$ in $T$ such that
  $\prev(T[i..i+k]) = \prev(T[b_k..b_k+k])$ for $0 \leq k \leq j-i$.
  By the definition of $\PPH(T)$,
  the root has an out-going edge labeled by $\prev(T[b_1+1..b_1+1])$,
  and this is the base case for our induction.
  Since $\prev(T[i..i+k]) = \prev(T[b_k..b_k+k])$,
  we have $\prev(T[i+1..i+k]) = \prev(T[b_k+1..b_k+k])$.
  Now since $\prev(T[b_{k+1}+1..b_{k+1}+k+1]) = \prev(T[i+1..i+k+1])$
  and $\prev(T[b_k+1..b_k+k]) = \prev(T[i+1..i+k])$,
  $\prev(T[b_k+1..b_k+k])$ is a prefix of $\prev(T[b_{k+1}+1..b_{k+1}+k+1])$.
  This implies that if $\prev(T[b_k+1..b_k+k])$ is represented by $\PPH(T)$,
  then $\prev(T[b_{k+1}+1..b_{k+1}+(k+1)])$ is also represented by $\PPH(T)$.
  By induction, we have that $\prev(T[b_{j-i}+1..b_{j-i}+j-i]) = \prev(T[i+1..j])$
  is represented by $\PPH(T)$.
  Applying the same argument inductively,
  it is immediate that $\prev(T[i+h...j])$ with $2 \leq h \leq j-i$
  are also represented by $\PPH(T)$.  
\end{proof}

In the next section,
we show how to construct our RL p-position heap
$\PPH(T)$ for an input text p-string $T$ of length $n$
in $O(n \log (\sigma + \pi))$ time and $O(n)$ space.

\section{Right to left construction of parameterized position heaps}

In this section,
we present our algorithm which constructs
$\PPH(T)$ of a given p-string $T$ in a right-to-left online manner.
The key to our construction algorithm is
the use of \emph{reversed suffix links}, which will be defined in the following subsection.

\subsection{Reversed suffix links}

For convenience,
we will sometimes identify each node $v$ of $\PPH(T)$ with the path label from the root to $v$.
In our right-to-left online construction of $\PPH(T)$,
we use the \emph{reversed suffix links},
which are a generalization of the \emph{Weiner links} that are used
in right-to-left construction of the suffix tree~\cite{Weiner}
for (standard) string matching:

\begin{definition}[Reversed suffix links]
\label{def:reversed_suffix_link}
For any node $v$ of $\PPH(T)$ and a character $a \in \Sigma \cup [0..n-1]$,
let
\[
 \rslink(a, v) =
 \begin{cases}
   av & \mbox{if $a \in \Sigma \cup \{0\}$ and $av$ is represented by $\PPH(T)$},\\
   u  & \begin{array}{l}
         \mbox{if $a \in [1..n-1]$, $v[a] = 0$ and} \\
         \mbox{$u = 0v[1..a-1]av[a+1..|v|]$ is represented by $\PPH(T)$},
        \end{array}
     \\
   \mbox{undefined} & \mbox{otherwise}.
 \end{cases}
\]
\end{definition}

\begin{figure}[tb]
   \centerline{
    \includegraphics[scale=0.65]{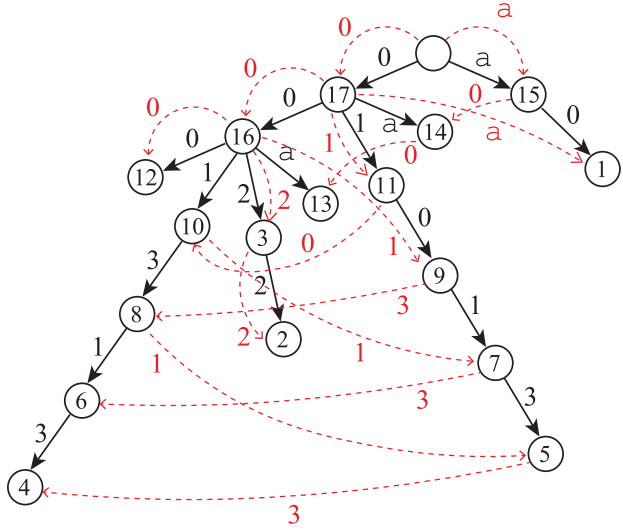}
  }
   \caption{Illustration of the reversed suffix links of
     $\PPH(T)$ with the same p-string $T = \mathtt{axyxyyxxyyxxzyazy}$
     as in Figure~\ref{fig:p-position_heap}. The reversed suffix links
     and their labels are shown in red.}
  \label{fig:suffix_link}
\end{figure}

It is clear that by taking one $\rslink$ link from a node,
then the node depth (and hence the string length) increases exactly one.

Observe that the first case of of the definition of $\rslink(a, v)$
is a direct extension of the Weiner links,
where $\rslink(a, v)$ points to the node $av$
that is obtained by prepending $a$ to $v$.
The second case, however, is a special case that arises
in parameterized pattern matching.
The following lemma ensures that our reversed suffix links $\rslink$ are well defined:
\begin{lemma} \label{lem:reversed_suffix_link_well_defined}
  For any node $v$ in $\PPH(T)$ and a character $a \in \Sigma \cup [0..n-1]$,
  let $\rslink(a, v) = u$, where $u$ is a node of $\PPH(T)$.
  Then, for any string $X$ such that $\prev(X) = u$,
  $\prev(X[2..|X|]) = v$.
\end{lemma}

\begin{proof}
  In the first case of the definition of $\rslink(a, v)$
  where $a \in \Sigma \cup \{0\}$, we have $\prev(X) = u = av$.
  Hence, $\prev(X[2..|X|]) = \prev(X)[2..|X|] = u[2..|u|] = v$.

  In the second case of the definition of $\rslink(a, v)$ where $a \in [1..n-1]$,
  we have $\prev(X) = u = 0v[1..a-1]av[a+1..|v|]$,
  which implies that $X[1] = X[a+1]$
  and $X[1] \neq X[i]$ for any $2 \leq i \leq a$.
  Thus, $\prev(X[2..|X|]) = v[1..a-1]0v[a+1..|v|] = v$.
\end{proof}

The next proposition shows
that there is a monotonicity in the labels of the reversed suffix links
that come from the nodes in the same path of $\PPH(T)$.

\begin{proposition} \label{prop:suffix_link_monotonicity}
  Suppose there is a reversed suffix link
  $\rslink(a, v)$ of a node $v$ with $a \in \Sigma \cup [0..n-1]$.
  Let $u$ be any ancestor of $v$.
  Then, if $a \in \Sigma \cup \{0\}$, $u$ has a reversed suffix link $\rslink(a, u)$.
  Also, if $a \in [1..n-1]$ and $|u| \geq a$, then $u$ has a reversed suffix link $\rslink(a, u)$,
  and if $a \in [1..n-1]$ and $|u| < a$, then $u$ has a reversed suffix link $\rslink(0, u)$.
\end{proposition}

\begin{proof}
  It suffices for us to show that the lemma holds for
  the parent $v'$ of $v$,
  since then the lemma inductively holds for any ancestor of $v$.
  Note that $v' = v[1..|v|-1]$.  Let $w = \rslink(a, v)$.

  If $a \in \Sigma \cup \{0\}$, then $w = av$.
  Hence, the parent of $w$ is $w[1..|w|-1] = av[1..|v|-1] = av'$.
  Therefore, there is a reversed suffix link $\rslink(a, v')$.

  If $a \in [1..n-1]$ and $|v'| = |v|-1 \geq a$,
  then it follows from the definition of $\rslink(a, v)$ that
  $v[a] = 0$ and $w = 0v[1..a-1]av[a+1..|v|]$. 
  Since $|v'| \geq a$, we have that $v'[a] = 0$ and $|v| \geq a+1$.  
  Thus $w[1..|w|-1] = 0v[1..a-1]av[a+1..|v|-1]$ is represented by $\PPH(T)$.
  Consequently, there is a reversed suffix link $\rslink(a, v')$.

  If $a \in [1..n-1]$ and $|v'| = |v|-1 = a-1$,
  then it follows from the definition of $\rslink(a, v)$ that
  $v[a] = v[|v|] = 0$ and $w = 0v[1..|v|-1]a$.
  Thus $w[1..|w|-1] = 0v[1..|v|-1] = 0v'$ is represented by $\PPH(T)$.
  Consequently, there is a reversed suffix link $\rslink(a, v')$.
\end{proof}

\subsection{Adding a new node}

Our algorithm processes 
a given p-string $T$ of length $n$ from right to left
and maintains $\PPH(T[i..])$ in decreasing order of $i = n, \ldots, 1$.
Initially, we begin with $\PPH(\varepsilon)$ which consists of the root $r$
representing the empty string $\varepsilon$.
For convenience, we use an auxiliary node $\bot$
as a parent of the root $r$,
and create reversed suffix links $\rslink(a, \bot) = r$
for every $a \in \Sigma \cup \{0\}$.

Now suppose we have constructed $\PPH(T[i..])$ for $1 < i \leq n$,
and we will update it to $\PPH(T[i-1..])$.
In so doing, we begin with node $v_i$ such that $\id(v_i) = i$.
We know the locus of this node $v_i$ since
$v_i$ is the node that was inserted at the last step
when $\PPH(T[i..])$ was constructed from $\PPH(T[i+1..])$.
Note also that this node $v_i$ is a leaf in $\PPH(T[i..])$.
We climb up the path from $v_i$ until finding
its lowest ancestor $v'_i$ that satisfies the following.
There are three cases:
\begin{enumerate}
 \item If $T[i-1] \in \Sigma$, then $v'_i$ is the lowest ancestor of $v_i$
       such that $\rslink(T[i-1], v_i)$ is defined.
 \item If $T[i-1] \in \Pi$ and $T[i-1] \neq T[j]$ for any $i \leq j \leq n$,
       then $v'_i$ is the lowest ancestor of $v_i$
       such that $\rslink(0, v_i)$ is defined.
 \item Otherwise,
       let $d = j-i$ where
       $j$ is the smallest position such that $i \leq j \leq n$ and $T[i-1] = T[j]$.
       Then $v'_i$ is the lowest ancestor of $v_i$
       such that $\rslink(d, v'_i)$ is defined if it exists,
       and $v'_i$ is the lowest ancestor of $v_i$ such that $\rslink(0, v'_i)$ is defined
       otherwise.
\end{enumerate}
Let $u_i$ be the node of $\PPH(T[i..])$
that is pointed by the reversed suffix link of $v'_i$ as above.
Then, we create a new node $v_{i-1}$ as a child of $u_i$
such that $\id(v_{i-1}) = i-1$.
The new edge $(u_i, v_{i-1})$ is labeled by $\prev(T[i-1..])[|u_i|+1]$.
We repeat the above procedure for all positions $i$ in $T$ in decreasing order.
See also Figure~\ref{fig:snapshot_const} for concrete examples.

\begin{figure}[tbp]
  \centerline{
    \includegraphics[scale=0.55]{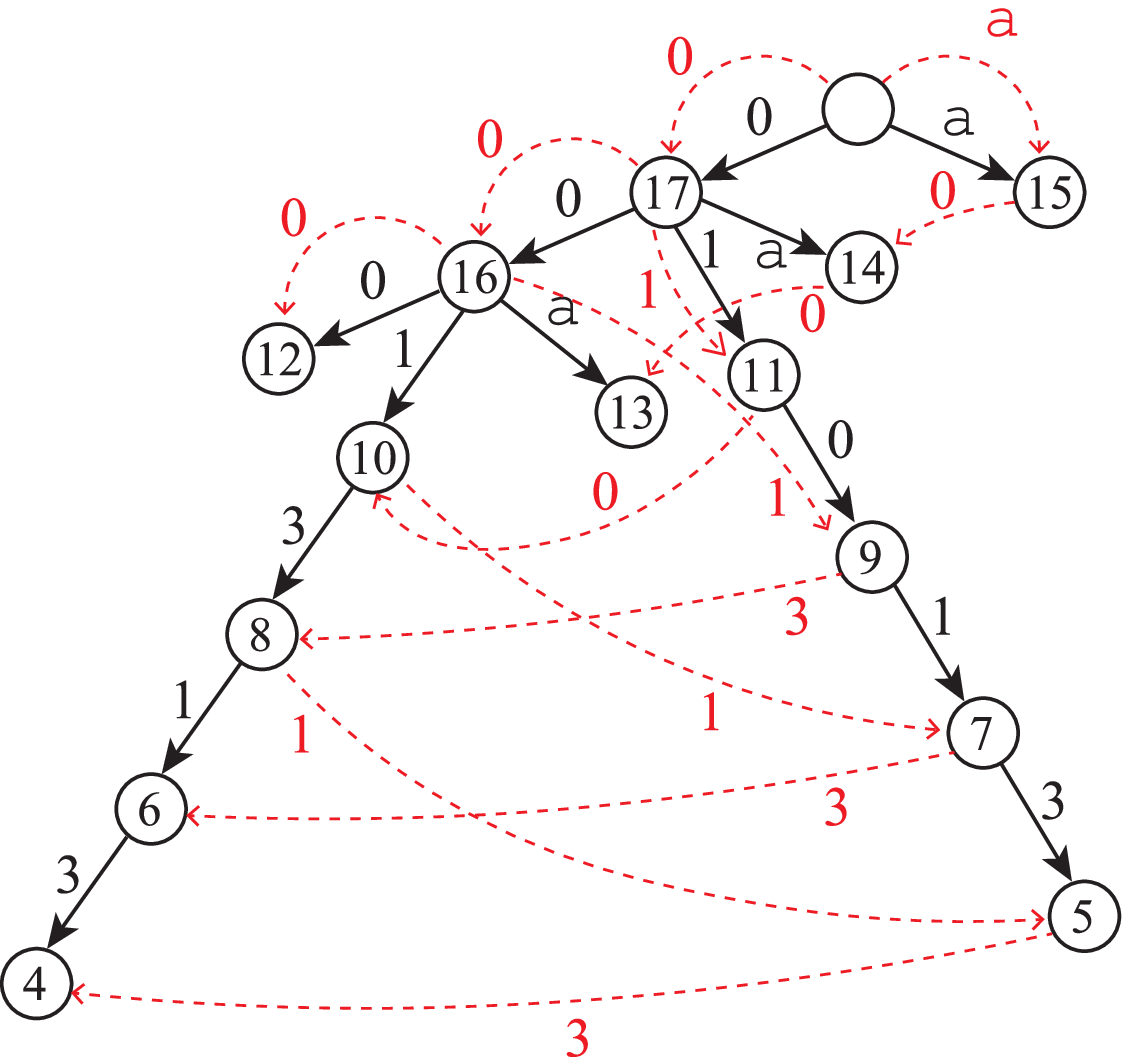}
    \hfill
    \includegraphics[scale=0.55]{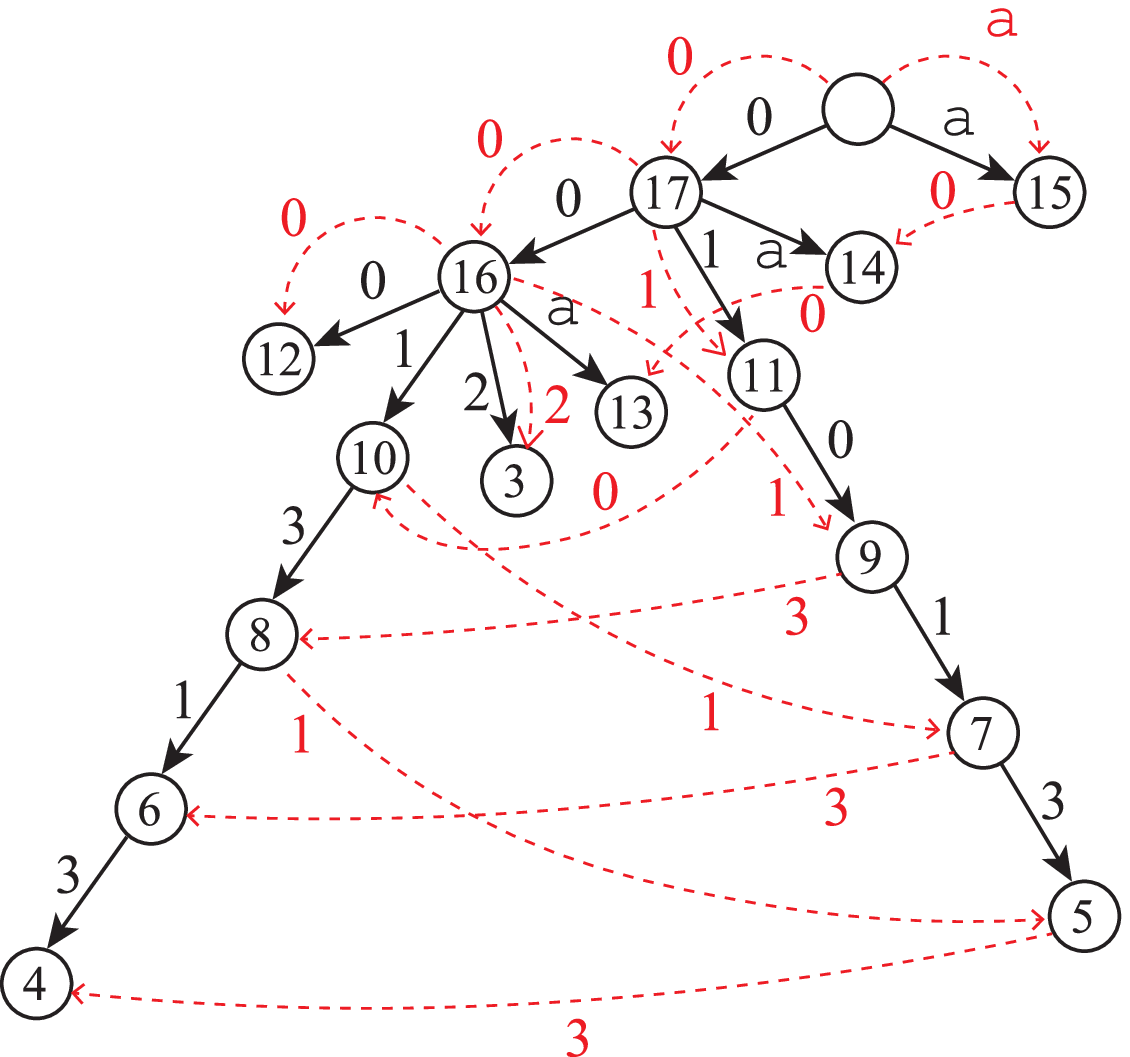}
  }
  \vspace*{5mm}
  \centerline{
    \includegraphics[scale=0.55]{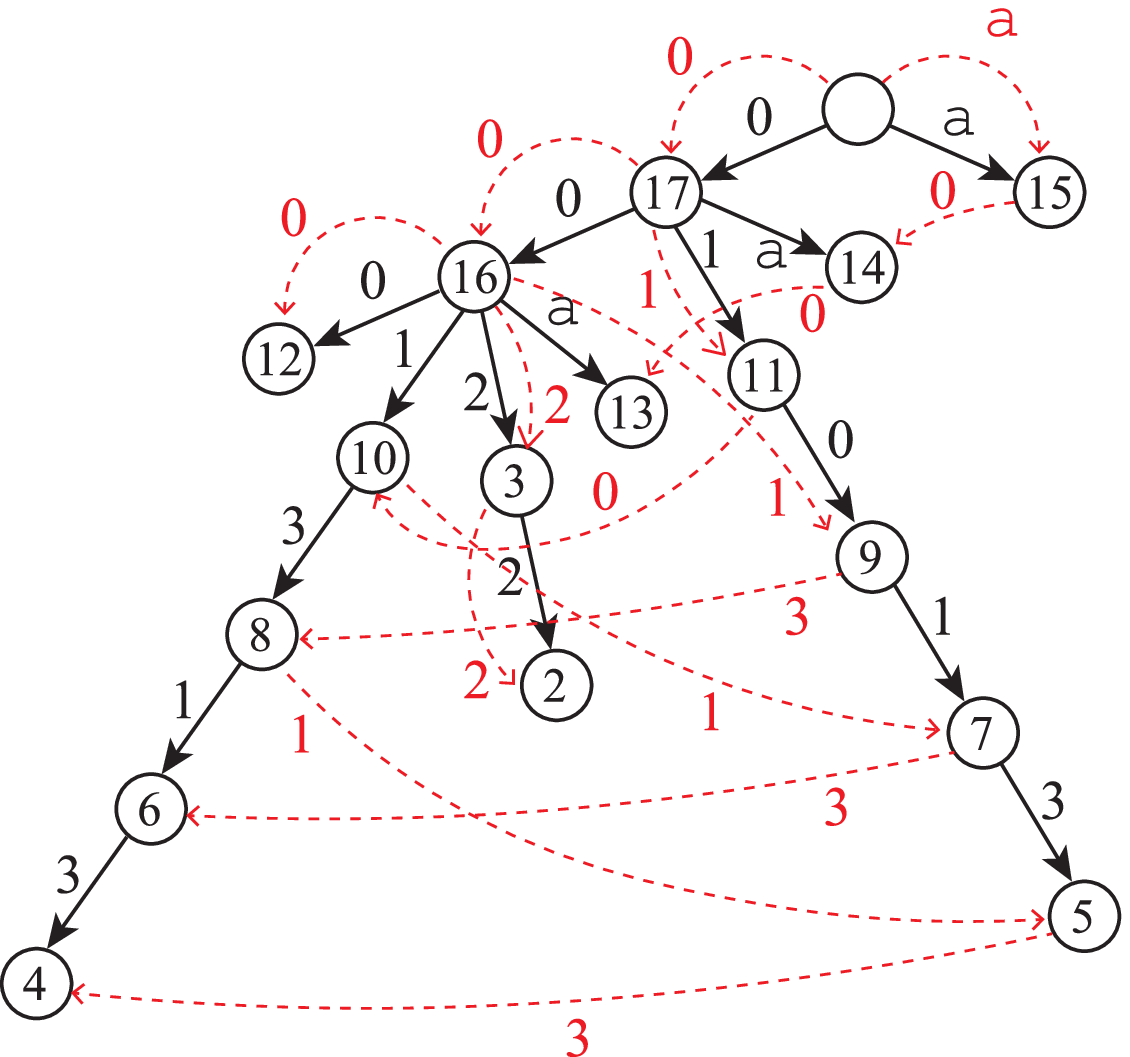}
    \hfill
    \includegraphics[scale=0.55]{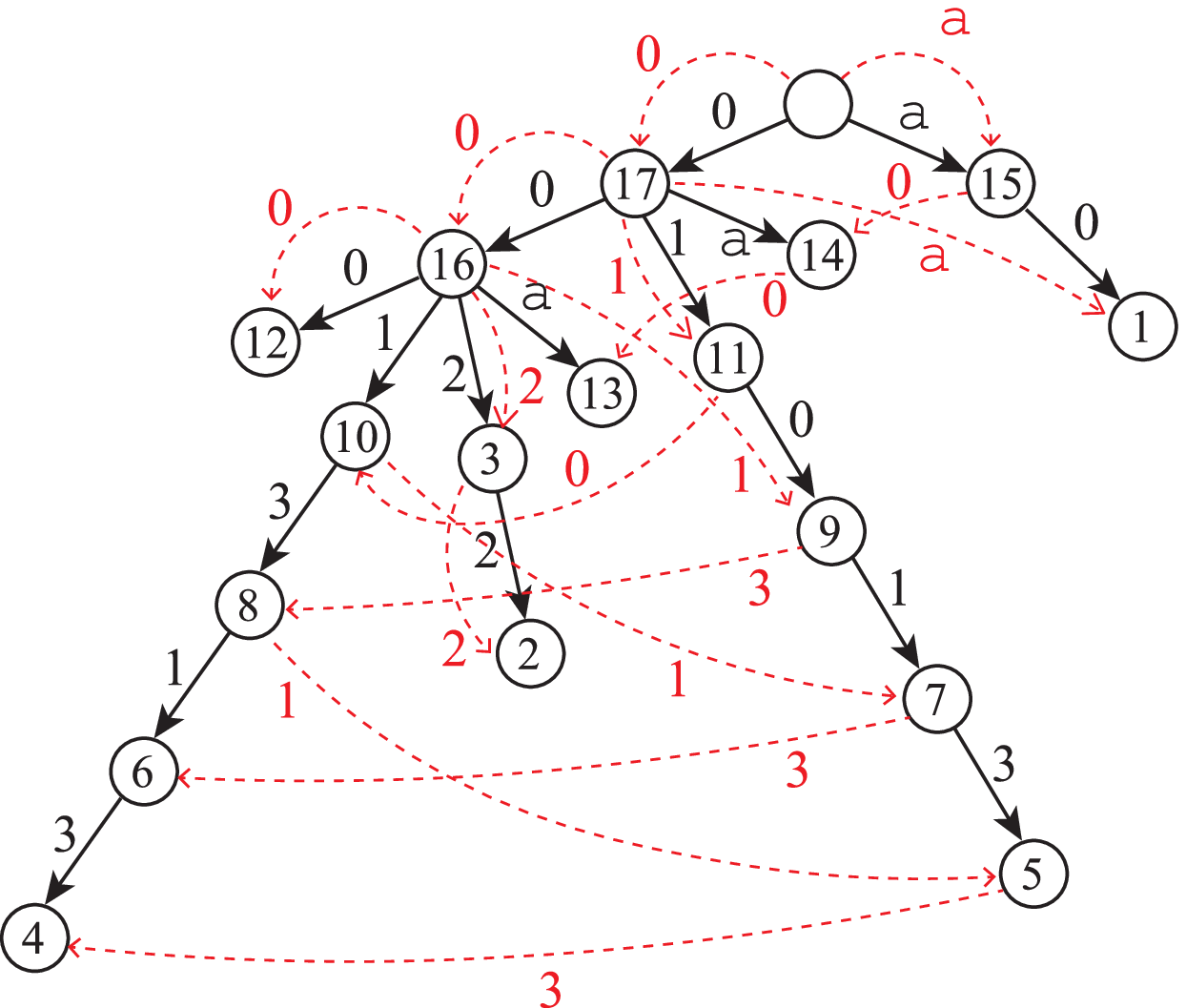}
  }
  \caption{A snapshot of updating $\PPH(T[i..])$ for $i = 4, 3, 2, 1$
    with the same p-string $T = \mathtt{axyxyyxxyyxxzyazy}$
    as in Figures~\ref{fig:p-position_heap} and \ref{fig:suffix_link}.
    First, we update $\PPH(T[4..])$ (upper left) to $\PPH(T[3..])$ (upper right).
    Since $T[3] = T[5] = \mathtt{y}$ and $d = 5-3 = 2$,
    we first try to find the lowest ancestor of the node with id $4$
    that has a reversed suffix link labeled with $d = 2$ by climbing up the path.
    However, it does not exist,
    and then we arrive at the lowest ancestor with id $17$ whose depth is $1$~($< 2$).
    Hence the second sub-case of Case 3 is applied,
    and using its reversed suffix link we move to the node with id $16$.
    The new node with id $3$ is inserted as its child.
    Next, we update $\PPH(T[3..])$ (upper right) to $\PPH(T[2..])$ (lower left).
    Since $T[2] = T[4] = \mathtt{x}$ and $d = 4-2 = 2$,
    we first try to find the lowest ancestor of the node with id $3$
    that has a reversed suffix link labeled with $d = 2$ by climbing up the path,
    and we arrive at the node with id $16$.
    Hence the first sub-case of Case 3 is applied,
    and using its reversed suffix link we move to the node with id $3$.
    The new node with id $2$ is inserted as its child.
    Finally, we update $\PPH(T[2..])$ (lower left) to $\PPH(T[1..])$ (lower right).
    Since $T[1] = \mathtt{a} \in \Sigma$, Case 1 is applied.
    Thus we try to find the lowest ancestor of the node with id $2$
    that has a reversed suffix link labeled with $\mathtt{a}$ by climbing up the path,
    and we arrive at the root.
    Using its reversed suffix link, we move to the node with id $15$.
    The new node with id $1$ is inserted as its child.
  }
  \label{fig:snapshot_const}
\end{figure}

\begin{lemma}
  The above algorithm correctly updates
  $\PPH(T[i..])$ to $\PPH(T[i-1..])$.
\end{lemma}

\begin{proof}
  Note that $v_i$ and $v'_i$ are prefixes of $\prev(T[i..])$.
  Let $a$ be the character in $\Sigma \cup [0..n-1]$
  that is used in the reversed suffix link as above.

  In Cases 1 and 2 above,
  we have $a = T[i-1] \in \Sigma$ or $a = 0$. 
  Then it is clear that $av'_i$ is a prefix of $\prev(T[i-1..])$.
  Since $v'_i$ is the lowest ancestor of $v_i$
  for which $\rslink(a, v'_i)$ is defined,
  $u_i = a v'_i$ is the longest prefix of $\prev(T[i-1..])$
  that is represented by $\PPH(T[i..])$.
  Hence, the new node $v_{i-1}$ and its incoming edge labeled by
  $\prev(T[i-1..])[|u_i|+1]$ are correctly inserted. 

  Consider Case 3 above.
  We first try to find $v'_i$ in the first sub-case, where $a = d \geq 1$.
  If it exists,
  then $v'_i$ is the lowest ancestor of $v_i$ such that $\rslink(d, v'_i)$ is defined,  
  and thus $\rslink(d, v'_i) = 0v'_i[1..d-1]dv'_i[d+1..|v'_i|]$.
  It now follows from Lemma~\ref{lem:substring_containment}
  that $u_i = 0v'_i[1..d-1]dv'_i[d+1..|v'_i|]$ is the longest prefix of
  $\prev(T[i-1..])$ that is represented by $\PPH(T[i..])$.
  Hence, the new node $v_{i-1}$ and its incoming edge labeled by
  $\prev(T[i-1..])[|u_i|+1]$ are correctly inserted in this sub-case.
  It is clear that $v'_i$ in the first sub-case is at least of depth $d$.
  Hence, if we arrive at the ancestor of $v_i$
  of depth $d-1$ without encountering the lowest ancestor
  satisfying the condition of the first sub-case,
  then we try to find the lowest ancestor of $v_i$
  that has a reversed suffix link labeled by $0$ (second sub-case).
  Thus, by a similar argument to Case 2,
  the new node $v_{i-1}$ its incoming edge labeled by
  $\prev(T[i-1..])[|u_i|+1]$ are correctly inserted in this second sub-case.
\end{proof}

\subsection{Adding a new reversed suffix link}

After inserting the new node $v_{i-1}$,
we need to maintain the reversed suffix links corresponding to $v_{i-1}$.

\begin{lemma} \label{lem:suffix_link_1}
  There is exactly one reversed suffix link that points
  to the new node $v_{i-1}$ in $\PPH(T[i-1..])$.
  Moreover, this reversed suffix link comes from
  the ancestor of $v_{i}$ of depth $|v'_{i}|+1$.
\end{lemma}

\begin{proof}
  Suppose on the contrary that
  there are two distinct nodes $x$ and $y$
  each of which has a reversed suffix link pointing to $v_{i-1}$.
  The label of any reversed suffix link that points to $v_{i-1}$
  is uniquely determined by the path label from the root to $v_{i-1}$.
  Therefore, 
  the reversed suffix links of $x$ and $y$ that point to $v_{i-1}$
  are both labeled by the same symbol.
  This means that $x = y$,
  however, this contradicts the definition of the p-position heap.
  Hence, there is at most one node which has a reversed suffix link
  that points to $v_{i-1}$.

  Let $z_i$ be the ancestor of $v_{i}$ of depth $|v'_{i}|+1$.
  Also, let $x = (T[i..])[|u_i|] = (T[i-1..])[|u_i|+1] = T[i+|u_i|-1]$,
  namely, $x$ is the text character that corresponds to
  the label of the edge $(v'_i, z_i)$ that is on the path from the root
  to
  \hbnote*{please check}{
  $v_{i}$,
  }
  and to the label of the new edge $(u_i, v_{i-1})$.
  If $x \in \Pi$ and $i+|u_i|-1$ is the smallest position in $T[i-1..]$
  such that $T[i-1] = T[i+|u_i|-1]$,
  then $(v'_i, z_i)$ is labeled with $0$
  while $(u_i, v_{i-1})$ is labeled with $|u_i|$.
  Otherwise,  the label of the new edge $(u_i, v_{i-1})$ must be equal to
  that of $(v'_i, z_i)$.
  It follows from the definition of reversed suffix links
  that in both cases the reversed suffix link to $v_{i-1}$ comes from $z_i$.  
\end{proof}

\begin{lemma} \label{lem:suffix_link_2}
  There is no reversed suffix link
  that comes from the new node $v_{i-1}$ in $\PPH(T[i-1..])$.
\end{lemma}

\begin{proof}
  Suppose on the contrary that there is a reversed suffix link from $v_{i-1}$
  in $\PPH(T[i-1..])$,
  and let $w$ be the node that is pointed by this reversed suffix link.
  Notice that $|w| = |v_{i-1}|+1$.
  Let $T[j..]$ be the suffix of $T$ for which this node $w$ was inserted,
  namely, $\id(w) = j > i-1$.
  By Lemma~\ref{lem:substring_containment},
  for any substring $X$ of $T[j..j+|w|-1]$,
  $\prev(X)$ is represented by $\PPH(T[j..])$,
  and hence it is also represented by $\PPH(T[i..])$ since $j \leq i$.
  Recall that $\prev(T[j+1..j+|w|-1]) = \prev(T[i-1..i+|v_{i-1}|])$,
  which implies that the node $v_{i-1}$ existed already in $\PPH(T[i..])$.
  However, this contradicts that $v_{i-1}$ is the node
  that was inserted when $\PPH(T[i..])$ was updated to $\PPH(T[i-1..])$.  
\end{proof}

Due to Lemmas~\ref{lem:suffix_link_1} and \ref{lem:suffix_link_2},
there is only one reversed suffix link that is newly inserted in $\PPH(T[i-1..])$.

\subsection{Complexity analysis}

\begin{lemma} \label{lem:construction_complexity}
  The proposed algorithm runs in a total of $O(n \log (\sigma + \pi))$ time
  with $O(n)$ space.
\end{lemma}

\begin{proof}
  For each $i = n, \ldots, 1$,
  the algorithm updates $\PPH(T[i..])$ to $\PPH(T[i-1..])$.
  The update begins with node $v_{i}$ such that $\id(v_{i}) = i$,
  and climbs up the path to $v'_{i}$.
  It takes a reversed suffix link from $v'_{i}$ and moves
  to $u_{i}$ of depth $|v'_{i}|+1$,
  and the new node $v_{i-1}$ of depth $|v'_{i}|+2$
  with $\id(v_{i-1}) = i-1$ is inserted.
  Hence the total number of nodes visited 
  when updating $\PPH(T[i..])$ to $\PPH(T[i-1..])$
  is $|v_i|-|v'_{i}|+2 = |v_i|-|v_{i-1}|+4$.
  Thus, the total number of nodes visited for all $i = n, \ldots, 1$
  sums up to $\sum_{i=n}^2 (|v_i|-|v_{i-1}| + 4) = |v_n|-|v_1| + 4(n-1) = O(n)$.
  At each node that we visit, it takes $O(\log (\sigma + \pi))$ time to search
  for the corresponding reversed suffix link,
  as well as inserting a new edge.
  Hence, the total time cost is $O(n \log (\sigma + \pi))$.

  It is clear that the number of nodes in $\PPH(T)$ is $n+2$,
  including the root and the auxiliary node $\bot$.
  It follows from Lemmas~\ref{lem:suffix_link_1} and \ref{lem:suffix_link_2}
  that the number of reversed suffix links coming out
  from the root, the internal nodes, and the leaves is $n+1$.
  As for the reversed suffix links that come from $\bot$ to the root,
  we add a new reversed suffix link labeled with $T[i-1]$
  only if $T[i-1] \in \Sigma$ and $T[i-1] \neq T[j]$ for any $j < i-1$.
  This way, we can maintain these reversed suffix links from
  $\bot$ in an online manner, using $O(n)$ space.
\end{proof}

We have proven the following theorem,
which is the main result of this paper.

\begin{theorem}
  For an input p-string $T$ of length $n$,
  the proposed algorithm constructs $\PPH(T[i..])$
  in a right-to-left online manner for $i = n, \ldots, 1$,
  in a total of $O(n \log (\sigma + \pi))$ time with $O(n)$ space.
\end{theorem}

\section{Parameterized pattern matching with augmented $\PPH(T)$}

Ehrenfeucht et al.~\cite{ehrenfeucht_position_heaps_2011}
introduced \emph{maximal reach pointers},
which used for efficient
pattern matching queries on position heaps.
Diptarama et al.~\cite{DiptaramaKONS17}
introduced maximal reach pointers for their LR p-position heaps,
and showed how to perform
pattern matching queries in $O(m \log (\sigma + \pi) + m\pi + \pocc)$ time,
where $m$ is the length of a given pattern p-string
and $\pocc$ is the number of occurrences to report.
We can naturally extend the notion of maximal reach pointers
to our RL p-position heaps, as follows:

\begin{definition}[Maximal reach pointers]
  For each position $1 \leq i \leq n$ in $T$,
  the \emph{maximal reach pointer} of the node $v$ with $\id(v) = i$
  points to the deepest node
  $u$ of $\PPH(T)$ such that $u$ is a prefix of $\prev(T[i..])$.
\end{definition}

We denote by $\mrp(i)$
the pointer of node $v$ such that $\id(v) = i$.
The \emph{augmented} $\PPH(T)$ is $\PPH(T)$
with the maximal reach pointers of all nodes.
For simplicity, if $\mrp(i)$ points to the node with id $i$,
then we omit this pointer.
See Figure~\ref{fig:maximal_reach_ponters}
for an example of maximal reach pointers and augmented $\PPH(T)$.

\begin{figure}[tb]
   \centerline{
    \begin{tabular}{|l|r|} \hline
     $\prev(T[12..])$ & $\uwave{\mathtt{a}}$ \\ \hline
     $\prev(T[11..])$ & $\uwave{0\mathtt{a}}$ \\ \hline
     $\prev(T[10..])$ & $\uwave{00\mathtt{a}}$ \\ \hline
     $\prev(T[9..])$ & $\uwave{\mathtt{a}00\mathtt{a}}$ \\ \hline
     $\prev(T[8..])$ & $\uwave{0\mathtt{a}03}\mathtt{a}$ \\ \hline
     $\prev(T[7..])$ & $\uwave{00\mathtt{a}3}3\mathtt{a}$ \\ \hline
     $\prev(T[6..])$ & $\uwave{\mathtt{a}00\mathtt{a}}33\mathtt{a}$ \\ \hline
     $\prev(T[5..])$ & $\uwave{0\mathtt{a}03}\mathtt{a}33\mathtt{a}$ \\ \hline
     $\prev(T[4..])$ & $\uwave{00\mathtt{a}3}3a33\mathtt{a}$ \\ \hline
     $\prev(T[3..])$ & $\uwave{\mathtt{a}00\mathtt{a}}33\mathtt{a}33\mathtt{a}$ \\ \hline
     $\prev(T[2..])$ & $\uwave{0\mathtt{a}03}\mathtt{a}33\mathtt{a}33\mathtt{a}$ \\ \hline
     $\prev(T[1..])$ & $\uwave{01}\mathtt{a}03\mathtt{a}33\mathtt{a}33\mathtt{a}$ \\ \hline
    \end{tabular}
    \hfil
    \raisebox{-35mm}{\includegraphics[scale=0.65]{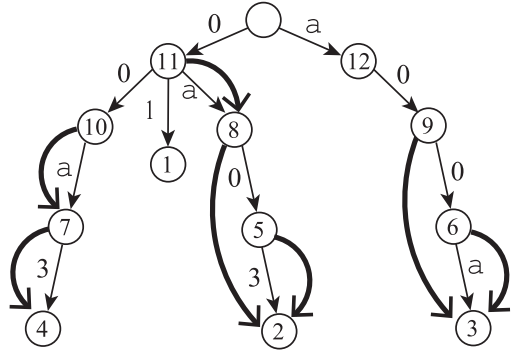}}
  }
   \caption{To the left is the list of $\prev(T[i..])$
     for p-string $T = \mathtt{xxayxayxayxa}$ of length $12$,
     where $\Sigma = \{\mathtt{a}\}$ and $\Pi = \{\mathtt{x}, \mathtt{y}\}$.
     To the right is an illustration for augmented $\PPH(T)$,
     where the maximal reach pointers are indicated by the bold arrows.
     The wavy underlined prefix of each $\prev(T[i..])$ in the left list
     denotes the longest prefix of $\prev(T[i..])$
     that is represented by $\PPH(T)$,
     and hence it is the destination of $\mrp(i)$.
}
  \label{fig:maximal_reach_ponters}
\end{figure}

\begin{lemma}
  For every $1 \leq i \leq n$,
  we can compute $\mrp(i)$ in a total of $O(n \log (\sigma + \pi))$ time
  with $O(n)$ space.
\end{lemma}

\begin{proof}
  We compute $\mrp(i)$ for each position $i = 1, \ldots, n$ increasing order.
  In so doing, we use the \emph{forward} suffix link
  that are the reversals of the reversed suffix links.
  For simplicity, we will call forward suffix links as suffix links.
  Since there is exactly one in-coming reversed suffix link
  to each node, there is also exactly one out-going suffix link
  from each node.
  Let $\slink(v)$ denote the node that the suffix link of $v$ points to.

  We begin with node $v_1$ such that $\id(v_1) = 1$.
  Since we have built $\PPH(T[i..])$ in decreasing order of $i$,
  $v_1$ is a leaf of $\PPH(T)$ and
  it is the deepest node that is a prefix of $\prev(T[1..])$.
  Now we take the suffix link of $v_1$,
  and let $u_1 = \slink(v_1)$.
  Since $\prev(T[1..|v_1|]) = v_1$,
  it follows from Lemma~\ref{lem:reversed_suffix_link_well_defined}
  that $u_1 = \prev(T[2..|v_1|])$,
  which implies that $u_1$ is a prefix of $\prev(T[2..])$.
  Then the deepest node $v_2$ that is a prefix of $\prev(T[2..])$
  can be found by traversing the corresponding path
  from node $u_1$.
  Then, we make a pointer to $v_2$ from the node $w$ with $\id(w) = 2$.
  We iteratively perform the same procedure for all positions $i$ in increasing order.
    
  To analyze the time complexity,
  we can use a similar argument as in Lemma~\ref{lem:construction_complexity}.
  For each $i$,
  the number of nodes traversed is $|v_{i+1}|-|u_i|+1 = |v_{i+1}|-|v_i|+2$.
  Thus, the total number of nodes visited
  sums up to $\sum_{i=1}^{n-1}(|v_{i+1}|-|v_i|+2) = |v_n|-|v_1| + 2(n-1) = O(n)$.
  Since it takes $O(\log (\sigma + \pi))$ time to search for
  each corresponding edge in the traversal,
  the total running time is $O(n \log (\sigma + \pi))$.

  The space requirement is clearly $O(n)$.
\end{proof}

It is straightforward that
by applying Diptarama et al.'s pattern matching algorithm to
our $\PPH(T)$ augmented with maximal reach pointers,
parameterized pattern matching can be done 
in $O(m \log (\sigma + \pi) + m \pi + \pocc)$ time.

\begin{corollary}
  Using our augmented $\PPH(T)$,
  one can perform parameterized pattern matching queries
  in $O(m \log (\sigma + \pi) + m\pi + \pocc)$ time.
\end{corollary}

\section{Conclusions and further work}

This paper proposed a new indexing structure for
parameterized pattern matching, called RL p-position heaps,
that are built in a right-to-left online manner.
We proposed a Weiner-type construction algorithm
for our RL p-position heaps that runs in $O(n \log (\sigma + \pi))$
time with $O(n)$ space,
for a given text p-string of length $n$
over a static alphabet $\Sigma$ of size $\sigma$
and a parameterized alphabet $\Pi$ of size $\pi$.
The key to our efficient construction is how to label
the reversed suffix links.
By augmenting our position heap with maximal reach pointers,
one can perform parameterized pattern matching
in $O(m \log(\sigma + \pi) + m \pi + \pocc)$ time,
where $m$ is the length of a query pattern and $\pocc$
is the number of occurrence to report.

Our future work includes the following:
\begin{itemize}
\item Would it be possible to shave
  the $m\pi$ term in the pattern matching time
  using parameterized position heaps?
  Other data structures such as parameterized suffix trees
  achieve better $O(m \log(\sigma + \pi) + \pocc)$ time~\cite{Baker96}.

\item Nakashima et al.~\cite{position_heaps_of_trie_2012}
  extended Ehrenfeucht et al.'s
  right-to-left position heaps~\cite{ehrenfeucht_position_heaps_2011}
  to a set of texts given as a trie.
  We are now working on extending our right-to-left p-position heaps
  to a set of texts given as a trie.
\end{itemize}

\bibliographystyle{abbrv}
\bibliography{ref}

\end{document}